\documentclass[sigconf,anonymous=False]{acmart}

\usepackage{subfig}
\usepackage{multirow}
\usepackage[normalem]{ulem}
\usepackage{algorithm}
\usepackage{algorithmic}

\usepackage{amsmath,amssymb,amsfonts}
\graphicspath{{./photo/}}
\useunder{\uline}{\ul}{}
\makeatletter

\newcommand{\Rmnum}[1]{\expandafter\@slowromancap\romannumeral #1@}
\makeatother

\AtBeginDocument{%
  \providecommand\BibTeX{{%
    \normalfont B\kern-0.5em{\scshape i\kern-0.25em b}\kern-0.8em\TeX}}}

\setcopyright{acmcopyright}
\copyrightyear{2018}
\acmYear{2018}
\acmDOI{XXXXXXX.XXXXXXX}

\acmConference[Conference acronym 'XX]{Make sure to enter the correct
  conference title from your rights confirmation emai}{June 03--05,
  2018}{Woodstock, NY}
\acmPrice{15.00}
\acmISBN{978-1-4503-XXXX-X/18/06}

\begin{document}

% \title{Potential Isomorphism Propagation: Advancing Weakly Supervised Entity Alignment}
\title{Understanding and Guiding Weakly Supervised Entity Alignment with Potential Isomorphism Propagation}


\author{Yuanyi Wang}
\affiliation{%
  \institution{Beijing University of Posts and
Telecommunications}
  \city{Beijing}
  \country{China}}
\email{wangyuanyi@bupt.edu.cn}

\author{Wei Tang}
\affiliation{%
  \institution{2012 Lab, Huawei Co. LTD}
  \city{Beijing}
  \country{China}}
\email{tangocean@bupt.edu.cn}

\author{Haifeng Sun}
\affiliation{%
  \institution{Beijing University of Posts and
Telecommunications}
  \city{Beijing}
  \country{China}}
\email{hfsun@bupt.edu.cn}

\author{Zirui Zhuang}
\affiliation{%
  \institution{Beijing University of Posts and
Telecommunications}
  \city{Beijing}
  \country{China}}
\email{zhaungzirui@bupt.edu.cn}

\author{Xiaoyuan Fu}
\affiliation{%
  \institution{Beijing University of Posts and
Telecommunications}
  \city{Beijing}
  \country{China}}
\email{fuxiaoyuan@bupt.edu.cn}

\author{Jingyu Wang}
\affiliation{%
  \institution{Beijing University of Posts and
Telecommunications}
  \city{Beijing}
  \country{China}}
\email{wangjingyu@bupt.edu.cn}

\author{Qi Qi}
\affiliation{%
  \institution{Beijing University of Posts and
Telecommunications}
  \city{Beijing}
  \country{China}}
\email{qiqi8266@bupt.edu.cn}

\author{Jianxin Liao}
\affiliation{%
  \institution{Beijing University of Posts and
Telecommunications}
  \city{Beijing}
  \country{China}}
\email{jxlbupt@gmail.com}

\renewcommand{\shortauthors}{Trovato and Tobin, et al.}

\begin{abstract}
Weakly Supervised Entity Alignment (EA) is the task of identifying equivalent entities across diverse knowledge graphs (KGs) using only a limited number of seed alignments. Despite substantial advances in aggregation-based weakly supervised EA, the underlying mechanisms in this setting remain unexplored. In this paper, we present a propagation perspective to analyze weakly supervised EA and explain the existing aggregation-based EA models. Our theoretical analysis reveals that these models essentially seek propagation operators for pairwise entity similarities. We further prove that, despite the structural heterogeneity across different KGs, the potentially aligned entities within aggregation-based EA models exhibit isomorphic subgraphs, a fundamental yet underexplored premise of EA. Leveraging this insight, we introduce a potential isomorphism propagation operator to enhance the propagation of neighborhood information across KGs. We develop a general EA framework, PipEA, incorporating this operator to improve the accuracy of every type of aggregation-based model without altering the learning process. Extensive experiments substantiate our theoretical findings and demonstrate PipEA's significant performance gains over state-of-the-art weakly supervised EA methods. Our work not only advances the field but also enhances our comprehension of aggregation-based weakly supervised EA.
% The code and data are available at \url{https://github.com/wyy-code/PipEA}.

\end{abstract}



%%删除一些自带的ACM格式内容
\renewcommand\footnotetextcopyrightpermission[1]{}
\settopmatter{printacmref=false} %remove ACM reference format



\maketitle

\section{Introduction}
\label{section:1}

\begin{figure}[t]
\centering
\includegraphics[width = \linewidth]{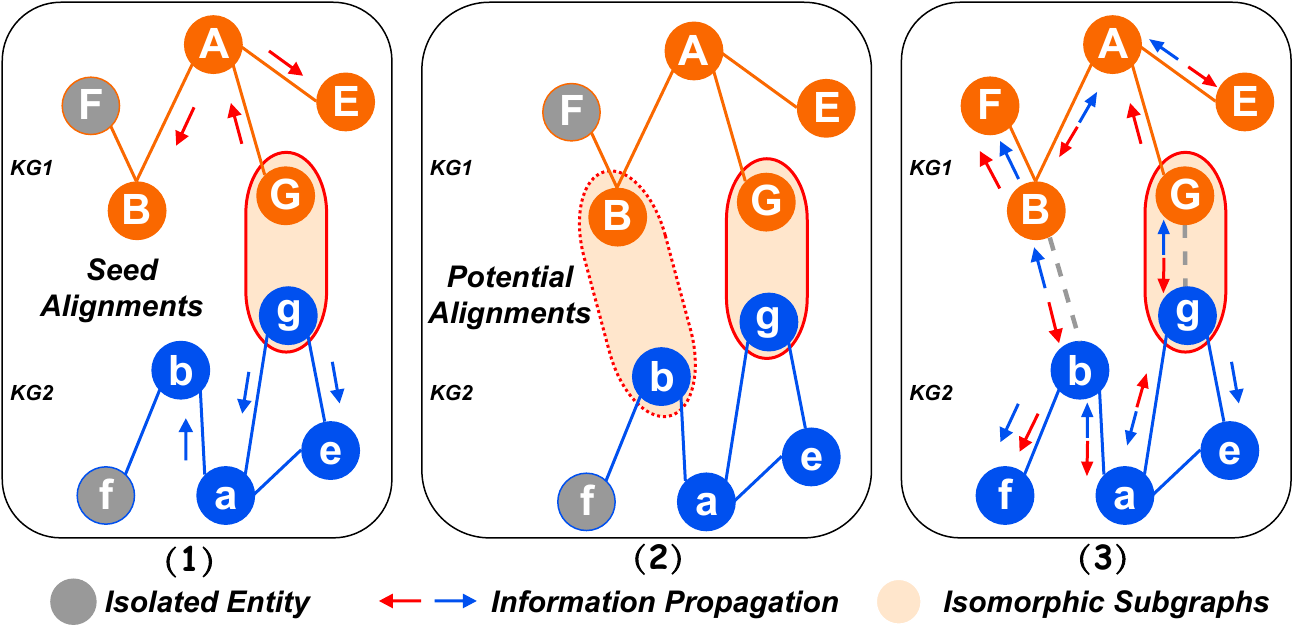}
\caption{(1) Propagation gap in weakly supervised settings and (3) our potential isomorphism propagation.} 
\label{figI}
\end{figure}

Knowledge Graphs (KGs) have emerged as pivotal resources across diverse domains, such as information retrieval \cite{yu2023retrieval}, question answering \cite{chakrabarti2022deep}, and recommendation systems \cite{wang2023mixed}. Despite their growing importance, KGs suffer from limitations in coverage, constraining their utility in downstream applications. The integration of heterogeneous KGs presents a significant challenge, at the core of which lies Entity Alignment (EA). EA aims to identify corresponding entities across different KGs. Contemporary EA solutions, particularly aggregation-based models, adhere to established pipelines. They rely on abundant seed alignments as supervised signals to learn entity representations, projecting diverse KGs into a unified embedding space, and subsequently predicting alignment results using these unified embeddings.

However, these methods heavily hinge on the availability of substantial seed alignments, which can be unrealistic or expensive to obtain. This has led to increased interest in weakly supervised EA, a scenario where only a limited number of seed alignments are accessible. Recent studies  \cite{ref_article15, ref_article16, ref_article33} have explored strategies like active learning and additional information incorporation to enhance the performance of aggregation-based EA models, which are state-of-the-art weakly supervised EA methods. While these efforts have shown some promise, it remains unclear why the adaptation of aggregation-based EA models for effective information propagation in weakly supervised settings is challenging.

This paper delves into Weakly Supervised EA (WSEA), focusing on how to enhance aggregation-based EA models in weakly supervised settings through the perspective of information propagation. We unveil that potentially aligned entities exhibit isomorphic subgraphs. We analyze the limitations of existing models in weakly supervised contexts, noting their reliance on local structural information and iterative neighbor aggregation. These models aim to minimize the distances between output embeddings of seed alignments, as further elaborated in Section \ref{analysis}. Figure \ref{figI} (1) depicts the key challenge: reliance on a sufficient quantity of seed alignments like (G,g). Limited seed alignments result in many unlabeled entities can not utilize the prior knowledge from labeled ones, hampering the propagation of alignment information due to restricted aggregation steps. For instance, in a two-layer GNN, only entities within a two-hop radius (like A, B) of seed entities participate in aggregation-based propagation, leaving distant entities isolated.

While the concept of employing higher-order aggregation-based models to expand neighborhood propagation has been considered, empirical studies \cite{PEEA} show their limited effectiveness in weakly supervised settings. Furthermore, research \cite{gasteiger2018predict} establishes a relationship between aggregation-based models and random walks. It demonstrates that, as the number of layers increases, these models converge towards a limiting distribution, which resembles that generated by random walks. Consequently, their performance deteriorates significantly with a high number of layers.

In light of these challenges, we conduct a theoretical analysis that unveils the learning process of aggregation-based models as a seek for propagation operators of pairwise entity similarities. Based on this insight, we establish a key theoretical result: \textit{potentially aligned entities in aggregation-based EA models possess isomorphic subgraphs, enabling the propagation of neighborhood information through these isomorphic subgraphs.} For example, as shown in Fig \ref{figI} right, potentially aligned entities (B,b) share isomorphic subgraphs enabling the propagation of neighborhood information through these isomorphic subgraphs. Leveraging this insight, we introduce \textbf{P}otential \textbf{i}somorphism \textbf{p}ropagation \textbf{E}ntity \textbf{A}lignment (PipEA), a general aggregation-based EA framework specifically designed to bridge the propagation gap in weakly supervised settings. PipEA constructs a propagation operator with two components: intra-graph propagation based on the original single-graph connectivity and inter-graph propagation grounded in potential alignment results represented by similarity matrices. This operator facilitates the generation of a new similarity matrix. We further propose a refinement scheme to better fuse the new and original similarity matrix. To reduce the complexity, we adopt randomized low-rank SVD \cite{ref_article41} and the sinkhorn operator \cite{cuturi2013sinkhorn}. Extensive experiments demonstrate PipEA's effectiveness, not only in weakly supervised settings but also in some normal supervised scenarios. In particular, our framework even improves the most dominant metric Hit@1 by nearly two times compared to the original model and also achieves state-of-the-art on both cross-lingual and mono-lingual datasets.

%, respectively. Additionally, the method enhanced with an iterative strategy achieves impressive boosts of 17.3\% and 9.4\% in Hits@1 on the 15KEN-FR and 100KDBP-Wiki datasets, respectively.

Our main contributions are summarized as follows:
\begin{itemize}
\item \textbf{Theoretical Analysis:}
We provide a theoretical analysis of aggregation-based EA models, revealing their reliance on propagation operators for deriving pairwise entity similarities. We also prove the existence of isomorphic subgraphs in potentially aligned entities, a fundamental yet underexplored aspect in EA.

\item \textbf{Isomorphism Propagation Operator:}
We introduce the novel propagation operator, which enables cross-graph information propagation, and prove its convergence.

\item \textbf{Innovative EA Framework:}
We propose PipEA, a theoretically grounded framework designed to enhance weakly supervised EA by facilitating neighborhood information propagation across heterogeneous graphs, the first in the field to our knowledge.

\item \textbf{Extensive Experiments:}
Experimental results validate our theoretical analysis and indicate our method achieves state-of-the-art on public datasets.
% The code is available at \url{https://anonymous.4open.science/r/PipEA-F12D}.
% \end{itemize}
The code is available at \url{https://github.com/wyy-code/PipEA}.
\end{itemize}

\section{Preliminary}

% We first introduce key EA concepts and notation. Then we discuss the basic aggregation-based models and introduce the \textit{feature propagation}. Finally, we introduce related work about our method.

\subsection{Problem Definition}
\noindent
% \newtheorem{myDef}{\noindent\bf Definition}[section]
% \begin{myDef}
% \textbf{A knowledge graph}, denoted as $\mathcal{G} = (\mathcal{E}, \mathcal{R}, \mathcal{T})$, comprises a set of entities $\mathcal{E}$, a set of relations $\mathcal{R}$, and a set of triples $\mathcal{T} = \{ (h, r, t) \,|\, h, t \in \mathcal{E}, r \in \mathcal{R} \}$. Each triple represents an edge from the head entity $h$ to the tail entity $t$ with the relation $r$.
% \end{myDef}
% \noindent
% \begin{myDef}
% \textbf{EA task} aims to discover a one-to-one mapping of entities $\Phi$ from a source KG $\mathcal{G}_s = (\mathcal{E}_s, \mathcal{R}_s, \mathcal{T}_s)$ to a target KG $\mathcal{G}_t = (\mathcal{E}_t, \mathcal{R}_t, \mathcal{T}_t)$. Formally, seed alignment is denoted as $\Phi = \{ (e_s, e_t) \,|\, e_s \in \mathcal{E}_s, e_t \in \mathcal{E}_t, e_s \equiv e_t \}$, where $\equiv$ represents an equivalence relation between $e_s$ and $e_t$.
% \end{myDef}

\newtheorem{myDef}{\noindent\bf Definition}
\begin{myDef}
\textbf{A Knowledge Graph (KG)}, denoted as $G = (\mathcal{E}, \mathcal{R}, \mathcal{T})$, comprises a set of entities $\mathcal{E}$, a set of relations $\mathcal{R}$, and a set of triples $\mathcal{T} = \{ (h, r, t) \,|\, h, t \in \mathcal{E}, r \in \mathcal{R} \}$. Each triple represents an edge from the head entity $h$ to the tail entity $t$ with the relation $r$.
\end{myDef}
\noindent
\begin{myDef}
\textbf{WSEA task} seeks a one-to-one mapping $\Phi = \{(e_i, e_i^\prime) | e_i \in \mathcal{E}_s, e_i^\prime \in \mathcal{E}_t, e_i \equiv e_i^\prime\}$ from the source graph $G_s = (\mathcal{E}_s, \mathcal{R}_s, \mathcal{T}_s)$ to the target graph $G_t = (\mathcal{E}_t, \mathcal{R}_t, \mathcal{T}_t)$, relying on only limited seed alignments $\Phi^\prime \in \Phi$. Here, $\equiv$ represents an equivalence relation between $e_i$ and $e_i^\prime$.
\end{myDef}

The notations used are summarized in the Appendix \ref{appendix:notation}.

\subsection{Related Work}
\noindent
\textbf{Aggregation-based EA.}
The adoption of aggregation-based models, featuring graph neural networks (GNNs), has gained significant traction in the domain of EA \cite{ref_article9,ref_article10,ref_article11}. These models harness the power of GNNs to generate entity representations by aggregating information from neighboring entities \cite{GNNbased}. Diverse GNN-based variants, such as RDGCN\cite{RDGCN}, RNM\cite{RNM}, KEGCN\cite{KECG}, MRAEA\cite{ref_article24}, and RREA\cite{ref_article25}, have emerged to address the capture of structural information and neighborhood heterogeneity. Some of these models focus on optimizing the proximity of positive entity pairs (e.g., PSR\cite{ref_article38}, Dual-AMN\cite{ref_article29}) or the distance between negative pairs (e.g., SEA\cite{SEA}, TEA\cite{TEA}). Furthermore, attribute-enhanced techniques incorporate entity attributes such as names and textual descriptions \cite{AttrE, COTSAE, AttrGNN} to enhance entity embeddings. Notably, ACK \cite{ACK} constructs an attribute-consistent graph to mitigate contextual gaps. Our work contributes by shedding light on these models' underlying principles, revealing their quest for pairwise similarity propagation operators, and justifying the existence of isomorphic subgraphs within potentially aligned entities. Additionally, we introduce a novel method designed to augment the performance of aggregation-based models.

\noindent
\textbf{Weakly Supervised EA.}
In scenarios with limited labeled data, many EA models suffer a drastic decline in alignment accuracy \cite{EAsurvey, wang2024towards, zhao2023weakly}. Some studies, such as ALEA \cite{ref_article15}, ActiveEA \cite{ref_article16}, and RCL, have explored reinforced active learning and the application of diverse heuristics to improve generalization \cite{ref_article32}. Contrastive learning approaches, exemplified by RAC, leverage unlabeled data as supervision signals \cite{zhao2023weakly}. Meanwhile, methods like PEEA employ anchor-aware positioning to effectively capture long dependencies, setting new benchmarks \cite{PEEA}. Despite these advancements, these methods do not address the propagation gap of the original models under this setting. Our method introduces a unique angle by focusing on cross-graph information propagation, providing a complementary enhancement to existing weakly supervised EA techniques without modifying the foundational model structure.

\section{Theoretical Analysis}
\label{analysis}
This section delves into the theoretical foundation of aggregation-based EA models and our proposed propagation operator, potential isomorphism propagation. The organization is as follows:

\begin{enumerate}
\item \textbf{Aggregation-based EA} introduces basic aggregation-based models, establishing the foundation for our Proposition \ref{pro1}.
\item \textbf{Propagation Operators} discusses the propagation operators, which are central to the aggregation-based models.
\item \textbf{Analysis of Aggregation-based EA} reveals that these models essentially seek propagation operators for pairwise entity similarities. Then it proves that the potentially aligned entities within aggregation-based EA models have isomorphic subgraphs.
\item \textbf{Potential Isomorphism Propagation}  proposes the novel theoretical propagation operator to facilitate information propagation cross graph and prove its convergence.
\end{enumerate}

\subsection{Aggregation-based EA}
In aggregation-based EA models, entities are initially represented by aggregating their neighbors in the unified space. For simplicity, we consider a one-layer GCN  \cite{GCN} with mean-pooling as the aggregation function. The entity embedding in GCNs and the primary objective of alignment learning can be expressed as follows:
\begin{equation}
\label{GCNembedding}
\mathbf{e}=\frac{1}{|N(e)|}{\sum_{e^\prime \in N(e)}} \mathbf{e^\prime}
\end{equation}
\begin{equation}
\label{obj}
\min\limits_{(\mathbf{e}_s, \mathbf{e}_t)\in \Phi} \mathrm{d}(\mathbf{e}_s, \mathbf{e}_t)
\end{equation}
where $\mathrm{d}(\cdot)$ represents a distance measure. The objective is to minimize the embedding distance between identical entities in seed alignments. While negative sampling methods aim to generate dissimilar entity pairs and train to distinguish the embeddings of dissimilar entities, Eq. \ref{obj} is the fundamental and commonly employed learning objective, which is the focus of our analysis.

\subsection{Propagation Operators}
The propagation operator is the core of propagation algorithms, governing information or influence spread in graphs. It is represented as a matrix or mathematical function. In general, propagation algorithms can be expressed as:
\begin{equation}
\label{propagation}
\pi_{u}(v) = \pi_{u}(v) P
\end{equation}
$\pi$ is used to signify the transition probability between entities. One example is personalized PageRank (PPR)  \cite{pagerank}, where the propagation operator $P$ often takes the form $P = D^{-1}A$. Here, $D$ is a diagonal matrix, with $D(i,j)$ representing entity $i$'s out-degree (for directed) or degree (for undirected graphs), and $A$ is the adjacency matrix. PPR $\pi_{u}(v)$ for entity $v$ regarding node $u$ quantifies the probability of a random walk with an $\alpha$ discount initiated from $u$ ending at $v$, with $\alpha$ denoting the probability of stopping at the current entity and $(1 - \alpha)$ the probability of transitioning to a random out-neighbor. PPR propagation can be expressed as:
\begin{equation}
\pi_{PPR} = \sum_{\ell=0}^{\infty} \alpha (1-\alpha)^\ell P^\ell
\end{equation}

\subsection{Analysis of Aggregation-based EA}
\label{Analysis of Aggregation-based EA}
\newtheorem{myTh}{\bf Proposition}[section]
\newtheorem{myTheo}{Theorem}[section]

The Aggregation-based EA model is commonly evaluated on heterogeneous graphs  \cite{ref_article26, uncertainty,xin2022large,ref_article18}. Entity similarities are computed using embeddings from aggregation-based models, denoted as $(x_1, x_2, \ldots, x_n)$ and $(y_1, y_2, \ldots, y_m)$, with $n=|\mathcal{E}_s|$ and $m=|\mathcal{E}_t|$. The pairwise similarity matrix $\Omega(i,j)$, crucial for entity alignment, is defined as:
\begin{equation}
\Omega(i,j) = (x_1; x_2; \ldots ; x_n)^\top(y_1; y_2; \ldots ; y_m)\in \mathbb{R}^{n\times m}
\end{equation}
Similar to most EA settings, we assume that each entity aligns with at most one entity in another KG.

\begin{myTh}
\label{pro1}
    In the Aggregation-based EA model, the primary objective is to derive a propagation operator governing pairwise entity similarities via embedding learning.
\end{myTh}
\begin{proof}
Please refer to Appendix \ref{ProofofProposition1}.
\end{proof}

Proposition \ref{pro1} clarifies that the Aggregation-based EA model seeks a propagation operator operating on the similarity matrix. Recent studies employ such kind operators on proximity matrices \cite{AROPE, NRP, STRAP}, representing entity proximity within a single graph. It suggests that the similarity matrix can also be seen as a specialized proximity matrix, where $\Omega (i, j)$ measures the proximity between entity $i$ and entity $j$ in another graph within a unified space. This implies that $\Omega (i, j)$ captures potential structural information between cross-graph entities in the unified space. Moreover, the propagation operator derived from the learning process unveils the structure of subgraphs for potentially aligned entity pairs within this unified space. Therefore, we provide the following proposition.

\begin{myTh}
\label{pro2}
    Let $\Lambda = [\lambda_1, \ldots, \lambda_n] \in \mathcal{R}^{n \times n}$ represent a matrix comprising $n$ arbitrary orthonormal vectors, signifying the propagation operator derived from aggregation-based EA models. Then potentially aligned entities have isomorphic subgraphs, and it follows that $Pr(rank(\Lambda\Lambda^\top)=n)= 1$.
\end{myTh}
\begin{proof}
Please refer to Appendix \ref{ProofofProposition2}.
\end{proof}

Since EA seeks a one-to-one mapping, the resulting similarity matrix must be orthonormal. When the similarity matrix $\Omega$ is considered as $\Lambda$, it encodes information about potentially aligned entities exhibiting isomorphic subgraphs, forming the basis for the convergence of the subsequent operator.

\subsection{Potential Isomorphism Propagation}
\label{PIP}

 This foundational observation underpins our introduction of the potential isomorphism propagation operator, which strategically utilizes the similarity matrix to control information propagation across entities through isomorphic subgraphs within diverse graphs. Different from existing operators, like PPR, that only propagate information within the graph, this propagation operator integrates both inter-graph and intra-graph propagation. The inter-graph propagation between potentially aligned entities across graph is governed by the similarity matrix, while intra-graph propagation relies on the graph's connectivity represented as $D^{-1}_\gamma A_\gamma$, where $D_\gamma$ and $A_\gamma$ are degree and adjacency matrices for graph $\mathcal{G}_\gamma,\gamma=\{s,t\}$. We propose the Potential Isomorphism Propagation operator $\Lambda_3$ as follows:
\begin{myTh}
\label{pro3}
    Let $\Lambda_1 = [D_s^{-1}A_s, \Omega] \in \mathbb{R}^{n \times (n+m)}$, $\Lambda_2 = [\Omega^\top, D_t^{-1}A_t] \in \mathbb{R}^{m \times (n+m)}$, and $\Lambda_3 = [\Lambda_1; \Lambda_2] \in \mathbb{R}^{(n+m) \times (n+m)}$. Consider $\Lambda_3 \in \mathbb{R}^{(n+m) \times (n+m)}$ as a symmetric graph operator with $\lambda_1, \ldots, \lambda_{d}$ as its $d$ dominant eigenvalues (in decreasing order of magnitude), where $|\lambda_i| > |\lambda_{i+1}|$, $1 \leq i \leq d$. Then the $\Lambda_3 \Lambda_3^\top$ converges to the ${d}$-dominant eigenvectors.
\end{myTh}
\begin{proof}
Please refer to Appendix \ref{ProofofProposition3}.
\end{proof}
This proposition proves the existence of our proposed propagation operator, $\Lambda_3$, by proving its convergence.

\begin{figure}[t]
\centering
\includegraphics[width = \linewidth]{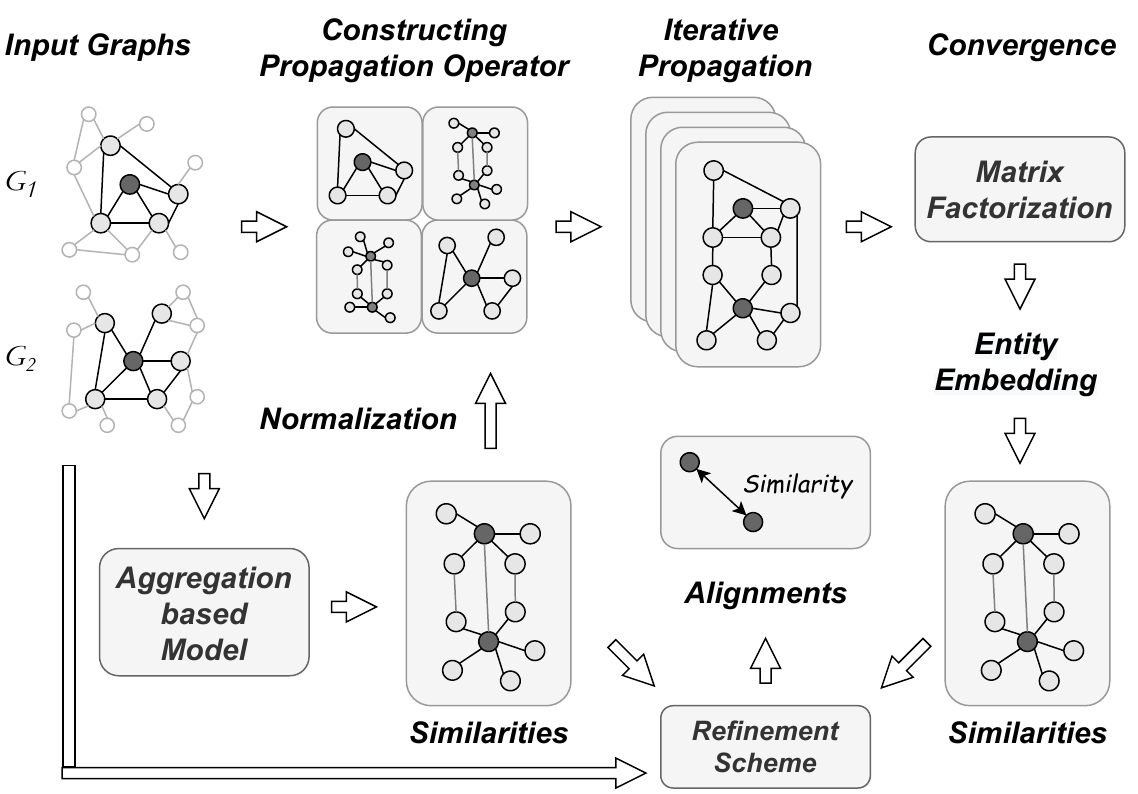}
\caption{PipEA initiates with pairwise similarities from aggregation-based models, employs potential isomorphism propagation to generate new similarities, and introduces a refinement scheme to effectively integrate them.} \label{figM}
\end{figure}

\section{Method}
\label{method}
This section presents the methodology of PipEA, a novel framework designed to enhance weakly supervised entity alignment through potential isomorphism propagation. Subsections are organized to sequentially detail each phase of our approach:

\begin{enumerate}
\item \textbf{Generating Initial Similarities:} This phase details the initial steps, including the training of encoders that solely rely on the structure and the initial similarities generation.
\item \textbf{Isomorphism Propagation:} At the core of PipEA, this phase leverages the principles established in Proposition \ref{pro3} to make potential isomorphism propagation across graphs. It includes the \textit{construction of the propagation operator}, \textit{propagation strategies}, \textit{matrix factorization}, and \textit{refinement schemes}.
\item \textbf{Predicting Alignments:} This phase redefines the one-to-one alignment prediction challenge as an assignment problem, optimizing the alignment process.
\item \textbf{Reducing Complexity:} Strategies to reduce complexity are implemented, enhancing the efficiency of the overall process.
\end{enumerate}
The details of PipEA are provided in Fig.~\ref{figM} and Algorithm \ref{AppAlgorithm}.

\subsection{Generating Initial Similarities}
PipEA initiates with the generation of initial pairwise similarities $\Omega_0$, utilizing aggregation-based EA models that focus on structural information, which is the most prominent information for EA \cite{EAsurvey}. We employ two leading EA models, DualAMN and PEEA, which are state-of-the-art (SOTA) for normal and weakly supervised EA, respectively. DualAMN enhances alignment accuracy through proxy matching and hard negative sampling within GCNs. In contrast, PEEA excels in structure-only scenarios by leveraging anchor positioning for effective dependency mapping. The similarities generation process is formalized as
\begin{equation}
    \Omega_0 = Encoder(A_s,A_t)
\end{equation}
where the $Encoder(\cdot)$ are the EA models to generate entity embedding and then calculate the pairwise similarities, and the inputs are only the structure data like adjacency matrix $A_s, A_t$ of KGs.

\subsection{Isomorphism Propagation}
\label{Propagation Scheme}
\subsubsection{Construction of Propagation Operator.}
Our propagation operator combines intra-graph and inter-graph propagation. Intra-graph propagation is based on the normalized adjacency matrix $D^{-1}A$ of the single-graph structure, while inter-graph propagation uses the similarity matrix $\Omega_0 \in \mathbb{R}^{n\times m}$. The operator is defined as:
\begin{equation}
\label{construction process}
    \Lambda_3=\begin{vmatrix} 
    \beta \cdot D_s^{-1}A_s & (1-\beta ) \cdot \mathcal{N}^{l_2}(\Omega_0)  \\
    (1-\beta ) \cdot \mathcal{N}^{l_2}(\Omega^\top_0) & \beta \cdot D_t^{-1}A_t
    \end{vmatrix}
    \in{\mathbb{R}^{(n+m)\times(n+m)}}
\end{equation}
Here, the parameter $\beta$ balances intra-graph and inter-graph propagation. Notably, we focus on similarities between highly confident potentially aligned pairs using a normalization operation denoted as $\mathcal{N}^{l_2}(\cdot)$, which is applied on rows as follows:
\begin{equation}
\label{normalization operations}
    \mathcal{N}^{l_2}(\Omega(i,:) = \frac{\varphi_k(\Omega(i,:))}{||\varphi_k(\Omega(i,:))||_2}
\end{equation}
where $\varphi_k$ denotes a ranking scheme that preserves the top $k$ candidates $(\omega_1^{c},\dots,\omega_k^{c})$ and sets others to zero.
\begin{equation}
\label{candidates}
    \mathcal{N}^{l_2}(\Omega(i,:) =[0,\dots,\omega_1^{c},\dots,0,\dots,\omega_k^{c},\dots,0]\in{\mathbb{R}^{m}}
\end{equation}
It's worth noting that for any $(e_j,e_{j^\prime})\in\Phi$ within the seed alignments, inter-graph propagation exclusively occurs between the aligned entity pairs, and their $\mathcal{N}^{l_2}(\Omega(j,:))$ is precisely defined as a vector with only one nonzero value at the $j^\prime$-th element:
\begin{equation}
\label{nonzero}
    \mathcal{N}^{l_2}(\Omega(j,:)= \mathbb{I}_{j^\prime}=[0,\dots,1,\dots,0]\in{\mathbb{R}^{m}}
\end{equation}

\subsubsection{Propagation Strategy.}
Inspired by the PPR formulation in Eq. \ref{propagation}, we introduce a random-walk propagation method to harness the potential isomorphism phenomenon, denoted as:
\begin{equation}
S = \sum_{\ell=0}^{\infty} \alpha (1-\alpha)^\ell \Lambda_3^\ell
\end{equation}
This method facilitates the propagation of neighborhood information through isomorphic subgraphs between potentially aligned entities. (shown in lines 4-7, Alg.\ref{AppAlgorithm})

\subsubsection{Matrix Factorization.}
\label{Matrix Factorization}
In experiments, we observed that matrix $S$ eliminates small values. To adapt for large-scale datasets, we introduce a threshold $\delta$ where values below it are set to zero. Then, we compute ${\log}(\frac{S}{\delta})$ for non-zero entries, obtaining a sparse matrix approximating the propagation results. The experimental results and analysis about $\delta$ are provided in Section \ref{Hyperparameters}. 

Next, using a differentiable Singular Value Decomposition (SVD)  \cite{ref_article41} with input dimension $d$, we factorize the matrix $S$. This produces $U$ and $V$ matrices, both of size $(n+m)\times d$, along with a diagonal matrix $\Sigma$, such that $U\Sigma{V^\top}\approx{S}$:
\begin{equation}
\label{lowranksvd}
U,\Sigma,{V^\top} = SVD (Sparse(S, \delta), d)
\end{equation}
Finally, we compute entity embeddings $X$ (in lines 11-13, Alg.\ref{AppAlgorithm})
\begin{equation}
    X = U\sqrt{\Sigma}
\end{equation}
This method ensures robust and informative information propagation within the unified space. Subsequently, we derive global pairwise similarities by using source embeddings $X_s=X[:n]$ and target embeddings $X_t=X[n:n+m]$:
\begin{equation}
    \Omega_0^\prime = X_s X_t^\top
\end{equation}

\subsubsection{Refinement Scheme.}
\label{Refining Scheme}
Two similarity matrices are generated in the previous process. The first is the initial similarity matrix produced by the aggregation-based EA model, emphasizing local information via n-hop neighborhood aggregation. The second is the propagation scheme, which focuses on global information in the unified space. To integrate these matrices, we employ element-wise Hadamard products, resulting in a new matrix denoted as $\overline{\Omega}_0$:
\begin{equation}
    \overline{\Omega}_0 = \Omega_0 \circ \Omega_0^\prime
\end{equation}
However, the direct fusion of these matrices can potentially compromise topological consistency information within the original similarity matrices. \cite{ref_article18}. To mitigate this loss, we introduce a refinement scheme that centers on preserving topological consistency. This refinement scheme leverages the concept of matched neighborhood consistency (MNC) scores \cite{ref_article42} to quantify topological consistency and iteratively enhances these scores.

The MNC score was originally defined in \cite{ref_article42} as:
\begin{equation}
    \mathbf{S}^{MNC} = A_s\Omega A_t\oslash(A_s\Omega \mathbf{1}^m\otimes\mathbf{1}^m + \mathbf{1}^nA_t\mathbf{1}^m-A_s\Omega A_t)
\end{equation}
Here, $\oslash$ denotes element-wise division, and $\otimes$ signifies the Kronecker product. However, for simplification purposes, we follow \cite{ref_article42} that approximates $\mathbf{S}^{MNC}$ as:
\begin{equation}
    \mathbf{S}^{MNC}\approx {A}_s\Omega {A}_t
\end{equation}
To iteratively update the similarity matrix $\overline{\Omega}$ in the $k$-th iteration, we introduce a small $\epsilon$ to every element of $\overline{\Omega}$ to assign every pair of entities a token match score, irrespective of whether the initial EA model identified them as matches. This enables us to rectify potential false negatives present in the initial similarity matrix. The similarities for aligned entity pairs are set as a one-nonzero-value vector, as illustrated in Eq. \ref{nonzero}:
\begin{equation}
\overline{\Omega}_k = \phi(\overline{\Omega}_{k-1}\circ A_s\overline{\Omega}_{k-1}A_t + \epsilon)
\end{equation}
Here, the function $\phi$ selects the entity pairs from the seed alignments $\Phi$ and assigns them a similarity score of 1, while setting all other values to zero. It is defined as:
\begin{equation}
\label{nonzero}
    \phi(\Omega, (e_i,e_j)\in \Phi) := \Omega(i)=\mathbb{I}_{j}=[0,\dots,1,\dots,0]
\end{equation}
This iterative refinement process is a critical component of our PipEA method (shown in lines 15-20, Alg.\ref{AppAlgorithm}), enabling the correction of potential biases and enhancing the alignment accuracy.

\noindent
\textbf{Remark:} The refinement Scheme is distinct from RefiNA \cite{ref_article42}. While RefiNA is an unsupervised graph matching technique, our scheme is supervised which needs seed alignments to update $\Omega_{k}$ and calculate the MNC score in each iteration.

\begin{algorithm}[t]
\caption{Potential Isomorphism Propagation Strategy}
\label{AppAlgorithm}
\renewcommand{\algorithmicrequire}{\textbf{Input:}}
\renewcommand{\algorithmicensure}{\textbf{Output:}}
\begin{algorithmic}[1] %[1] enables line numbers
\REQUIRE 
The adjacency matrices $A_s, A_t$, number of iterations $L_1,L_2$, embedding dimension $d$, seed alignments $\Phi$\\
\ENSURE 
The final refined similarity matrix $\overline{\Omega}$.\\
\STATE Initialize $S=\textbf{0}$, $R=\textbf{I}$. ($\textbf{I}$ is the identity matrix)
\STATE $\Omega_0 \leftarrow $ $Encoder(A_s, A_t)$
\STATE Constructing the propagation operator $\Lambda_3$.
\FOR{$k=1\rightarrow L_1$}
% \STATE \textcolor{gray}{\small $\triangleright \text{The stopping probability at each step to be $\alpha$}$}
\STATE $S \leftarrow S+{{\alpha}{\cdot}R}$
% \STATE \textcolor{gray}{\small $\triangleright \text{$(1-\alpha)$ probability to randomly jump}$}
% \STATE \textcolor{gray}{\small $\triangleright \text{to one of its out-neighbors.}$}
\STATE $R \leftarrow S+{(1-{\alpha}){\cdot}\Lambda_3{\cdot}R}$
\ENDFOR
\FOR{${\forall}S(i,j){\in}S$}
\STATE \textbf{if} $S(i,j) < \delta$,\quad \textbf{then} $S(i,j) \leftarrow 0$\\
\ENDFOR
\STATE Get matrix $\log (\frac{S}{\delta})$ for non-zero entries 
\STATE ${[U,\Sigma,{V^\top}]} \leftarrow$  Differentiable Sparse SVD($\log (\frac{S}{\delta}),d$)
\STATE Get the eigenvector entity embedding $X_s, X_t \leftarrow U\sqrt\Sigma$
\STATE $\overline{\Omega}^\prime \leftarrow $ $X_s X_t^\top$
\STATE $\overline{\Omega}_0 \leftarrow \Omega_0 \circ \Omega_0^\prime$
\FOR{$k=1\rightarrow L_2$}
\STATE $\overline{\Omega}_{k-1} \leftarrow \phi(\overline{\Omega}_{k-1}, \Phi)$
\STATE $\overline{\Omega}_k \leftarrow \overline{\Omega}_{k-1}\circ A_s\overline{\Omega}_{k-1}A_t + \epsilon$
\STATE $\overline{\Omega}_k \leftarrow$ Normalize $ \overline{\Omega}_k$ by row then column
\ENDFOR
\STATE \textbf{return} $\overline{\Omega}$
\end{algorithmic}
\end{algorithm}

\begin{table*}[t]
\caption{Main results on cross-lingual datasets with 1\% seed alignments. PEEA, Dual-AMN, and LightEA represent the SOTA in weakly supervised, normally supervised, and unsupervised EA. "Improv." indicates the percentage improvement over the original model. PipEA(D), PipEA(P), and PipEA(L) denote PipEA with Dual-AMN, PEEA, and LightEA as encoders, respectively.}
\label{tab1}
\setlength{\tabcolsep}{2pt}
\centering
\small
\resizebox{\textwidth}{!}{
\begin{tabular}{cc|ccccccccc|ccccccccc}
\toprule
\multicolumn{2}{c|}{\multirow{2}{*}{Datasets}}  & \multicolumn{9}{c|}{Cross-Lingual Datasets} & \multicolumn{9}{c}{Mono-Lingual Datasets}   \cr
\cmidrule{3-20}
\multicolumn{2}{c|}{}                           & \multicolumn{3}{c}{15KEN-DE}                      & \multicolumn{3}{c}{15KEN-FR}                        & \multicolumn{3}{c|}{100KEN-FR}                       & \multicolumn{3}{c}{15KDBP-Wiki}                    & \multicolumn{3}{c}{15KDBP-Yago}                  & \multicolumn{3}{c}{100KDBP-Wiki}                    \\
\cmidrule(lr){1-2}\cmidrule(lr){3-5}\cmidrule(lr){6-8}\cmidrule(lr){9-11}\cmidrule(lr){12-14}\cmidrule(lr){15-17}\cmidrule(lr){18-20}
\multicolumn{2}{c|}{Models}                     & H@1             & H@10           & MRR            & H@1             & H@10            & MRR             & H@1             & H@10            & MRR             & H@1             & H@10           & MRR             & H@1            & H@10           & MRR            & H@1             & H@10            & MRR             \\
\midrule
\multirow{13}{*}{\rotatebox{90}{Basic}}     & \multicolumn{1}{|c|}{GCN-Align}        & 10.9            & 26.7           & 16.4           & 3.6             & 15.2            & 7.1             & 2.5             & 9.4             & 5.0             & 3.1             & 11.0           & 5.8             & 40.1           & 60.6           & 47.1           & 3.5             & 11.4            & 6.2             \\
                            & \multicolumn{1}{|c|}{PSR}              & 21.5            & 49.7           & 31.0           & 15.1            & 38.1            & 22.9            & 13.2            & 32.9            & 19.9            & 19.5            & 44.2           & 27.9            & 25.3           & 51.6           & 34.2           & 14.6            & 33.5            & 21.0            \\
                            & \multicolumn{1}{|c|}{MRAEA}            & 28.6            & 58.7           & 38.7           & 14.4            & 38.5            & 22.4            & 13.5            & 36.1            & 21.0            & 19.4            & 45.4           & 28.1            & 42.9           & 72.4           & 53.1           & 17.1            & 39.7            & 24.7            \\
                            & \multicolumn{1}{|c|}{RREA}             & 48.5            & 72.5           & 56.8           & 26.3            & 56.4            & 36.4            & 16.4            & 40.6            & 24.5            & 41.8            & 67.5           & 50.7            & 82.1           & 92.8           & 86.0           & 21.4            & 45.9            & 29.7            \\
\cmidrule{2-20}
                            & \multicolumn{1}{|c|}{DualAMN}    & 51.9             & 75.4            & 60.1            & 25.2             & 52.1              & 34.3            & 15.0             & 38.6             & 22.8             & 40.0             & 64.5             & 48.5             & 76.2             & 88.3             & 80.7             & 16.2             & 37.7             & 23.5            \\
                            & \multicolumn{1}{|c|}{\textbf{PipEA(D)}} & \textbf{82.3}    & \textbf{86.4}   & \textbf{83.9}   & \textbf{48.5}    & \textbf{58.9}     & \textbf{52.4}   & \textbf{23.8}    & \textbf{43.7}    & \textbf{29.7}    & \textbf{71.6}    & \textbf{76.3}    & \textbf{73.4}    & \textbf{96.7}    & \textbf{97.8}    & \textbf{97.2}    & \textbf{31.9}    & \textbf{54.6}    & \textbf{39.6}   \\
                            & \multicolumn{1}{|c|}{\textbf{Improv.}}  & \textbf{58.6\%}  & \textbf{14.6\%} & \textbf{39.6\%} & \textbf{92.5\%}  & \textbf{13.1\%}   & \textbf{52.8\%} & \textbf{58.7\%}  & \textbf{13.2\%}  & \textbf{30.3\%}  & \textbf{79.0\%}  & \textbf{18.3\%}  & \textbf{51.3\%}  & \textbf{26.9\%}  & \textbf{10.8\%}  & \textbf{20.4\%}  & \textbf{96.9\%}  & \textbf{44.8\%}  & \textbf{68.5\%} \\
\cmidrule{2-20}
                            & \multicolumn{1}{|c|}{PEEA}             & 68.6             & 88.1            & 75.4            & 44.0             & 72.4              & 53.6            & 20.3             & 47.6             & 29.4             & 59.4             & 81.2             & 67.1             & 92.6             & 97.4             & 94.4             & 24.4             & 50.6             & 33.2            \\
                            & \multicolumn{1}{|c|}{\textbf{PipEA(P)}} & \textbf{85.4}    & \textbf{92.0}   & \textbf{87.8}   & \textbf{58.1}    & \textbf{74.4}     & \textbf{63.7}   & \textbf{32.7}    & \textbf{49.1}    & \textbf{38.5}    & \textbf{75.4}    & \textbf{82.5}    & \textbf{77.9}    & \textbf{96.9}    & \textbf{98.4}    & \textbf{97.5}    & \textbf{37.2}    & \textbf{58.5}    & \textbf{44.3}   \\
                            & \multicolumn{1}{|c|}{\textbf{Improv.}}  & \textbf{24.5\%}  & \textbf{4.4\%}  & \textbf{16.4\%} & \textbf{32.0\%}  & \textbf{2.8\%}    & \textbf{18.8\%} & \textbf{61.3\%}  & \textbf{3.1\%}   & \textbf{31.1\%}  & \textbf{26.9\%}  & \textbf{1.6\%}   & \textbf{16.1\%}  & \textbf{4.6\%}   & \textbf{1.0\%}   & \textbf{3.3\%}   & \textbf{52.5\%}  & \textbf{15.6\%}  & \textbf{33.4\%} \\                
\cmidrule{2-20}
& \multicolumn{1}{|c|}{LightEA}           & 59.1             & 0.8             & 65.6            & 25.3             & 50.8              & 33.8            & 27.0             & 51.8             & 35.3             & 41.1             & 67.6             & 50.1             & 69.6             & 86.1             & 75.3             & 32.9             & 55.7             & 40.69           \\
                            & \multicolumn{1}{|c|}{\textbf{PipEA(L)}} & \textbf{82.9}    & \textbf{89.4}   & \textbf{85.3}   & \textbf{50.9}    & \textbf{68.3}     & \textbf{56.9}   & \textbf{32.7}    & \textbf{49.1}    & \textbf{38.5}    & \textbf{74.9}    & \textbf{81.9}    & \textbf{77.4}    & \textbf{96.5}    & \textbf{98.1}    & \textbf{97.1}    & \textbf{38.3}    & \textbf{57.5}    & \textbf{44.30}  \\
                            & \multicolumn{1}{|c|}{\textbf{Improv.}}  & \textbf{23.8\%}  & \textbf{11.8\%} & \textbf{19.7\%} & \textbf{25.6\%}  & \textbf{17.5\%}   & \textbf{23.1\%} & \textbf{21.1\%}  & \textbf{-5.3\%}  & \textbf{9.1\%}   & \textbf{33.80\%} & \textbf{14.30\%} & \textbf{27.30\%} & \textbf{26.90\%} & \textbf{12.00\%} & \textbf{21.80\%} & \textbf{16.5\%}  & \textbf{3.2\%}   & \textbf{8.9\%}  \\
\midrule
\multirow{15}{*}{\rotatebox{90}{Semi-supervised}} & \multicolumn{1}{|c|}{BootEA}           & 0.6             & 3.6            & 1.7            & 2.7             & 10.1            & 5.2             & 2.1             & 4.4             & 5.4             & 1.8             & 7.4            & 3.7             & 2.7            & 1.6            & 6.6            & 3.3             & 40.5            & 28.2            \\
                            & \multicolumn{1}{|c|}{KECG}             & 42.5            & 64.6           & 50.2           & 14.1            & 43.3            & 23.7            & 11.1            & 30.5            & 17.7            & 23.8            & 45.8           & 31.3            & 57.8           & 78.8           & 65.1           & 20.2            & 42.4            & 27.8            \\
                            & \multicolumn{1}{|c|}{SEA}              & 43.1            & 66.5           & 51.2           & 18.9            & 49.4            & 29.1            & 12.5            & 34.5            & 19.9            & 15.6            & 40.4           & 24.0            & 81.4           & 92.7           & 85.5           & 13.6            & 33.6            & 20.4            \\
                            & \multicolumn{1}{|c|}{PSR}              & 79.9            & 91.4           & 84.1           & 52.8            & 75.3            & 60.5            & 55.4            & 72.6            & 62.0            & 72.1            & 85.5           & 77.1            & 95.2           & 97.9           & 96.3           & 59.3            & 68.4            & 61.7            \\
                            & \multicolumn{1}{|c|}{MRAEA}            & 64.7            & 84.5           & 72.3           & 35.9            & 61.7            & 44.7            & 38.2            & 64.1            & 55.1            & 58.6            & 78.4           & 66.5            & 88.5           & 97.8           & 92.7           & 45.8            & 60.2            & 48.4            \\
                            & \multicolumn{1}{|c|}{RREA}             & 76.5            & 90.7           & 81.8           & 39.6            & 68.0            & 49.5            & 55.2            & 74.1            & 63.3            & 66.8            & 83.5           & 73.1            & 95.6           & 98.3           & 96.7           & 58.0            & 71.8            & 62.7            \\
\cmidrule{2-20}
                            & \multicolumn{1}{|c|}{Dual-AMN}          & 77.1            & 93.0           & 83.6           & 48.4            & 79.0            & 59.4            & 57.3            & 76.2            & 65.9            & 67.2            & 86.6           & 75.1            & 92.4           & 98.4           & 95.2           & 59.6            & 72.1            & 63.8            \\
                            & \multicolumn{1}{|c|}{\textbf{PipEA(D)}} & \textbf{86.6}    & \textbf{93.5}   & \textbf{88.1}   & \textbf{61.3}    & \textbf{70.6}   & \textbf{64.9}   & \textbf{67.7}    & \textbf{84.9}    & \textbf{73.3}    & \textbf{78.4}    & \textbf{82.5}   & \textbf{80.0}   & \textbf{97.0}   & \textbf{98.5}   & \textbf{97.6}   & \textbf{71.8}    & \textbf{70.6}   & \textbf{68.1}   \\
                            & \multicolumn{1}{|c|}{\textbf{Improv.}}  & \textbf{12.26\%} & \textbf{0.52\%} & \textbf{5.41\%} & \textbf{26.74\%} & -10.61\%      & \textbf{9.33\%} & \textbf{18.15\%} & \textbf{11.40\%} & \textbf{11.29\%} & \textbf{16.61\%} & -4.70\%      & \textbf{6.48\%} & \textbf{1.56\%} & \textbf{0.1\%}      & \textbf{0.89\%} & \textbf{20.45\%} & -2.08\%     & \textbf{6.68\%} \\
\cmidrule{2-20}
                            & \multicolumn{1}{|c|}{PEEA}             & 83.6             & 93.2            & 87.1            & 55.6             & 79.4              & 64.0            & 59.6             & 77.6             & 66.3             & 78.5             & 90.0             & 82.7             & 96.6             & 98.6             & 97.3             & 65.3             & 78.2             & 70.6            \\
                            & \multicolumn{1}{|c|}{\textbf{PipEA(P)}}    & \textbf{90.7}   & \textbf{95.4}  & \textbf{92.5}  & \textbf{80.1}   & \textbf{89.8}   & \textbf{83.6}   & \textbf{71.6}   & \textbf{86.1}   & \textbf{76.7}   & \textbf{84.3}   & \textbf{90.9}  & \textbf{86.1}   & \textbf{97.2}  & \textbf{98.9}  & \textbf{97.7}  & \textbf{77.6}   & \textbf{86.8}   & \textbf{80.8}   \\
                            & \multicolumn{1}{|c|}{\textbf{Improv.}} & \textbf{8.5\%}  & \textbf{2.4\%} & \textbf{6.2\%} & \textbf{44.1\%} & \textbf{13.1\%} & \textbf{30.6\%} & \textbf{20.1\%} & \textbf{11.0\%} & \textbf{15.7\%} & \textbf{7.4\%}  & \textbf{1.0\%} & \textbf{4.1\%}  & \textbf{0.6\%} & \textbf{0.3\%} & \textbf{0.4\%} & \textbf{18.8\%} & \textbf{11.0\%} & \textbf{14.4\%}\\
\cmidrule{2-20}
& \multicolumn{1}{|c|}{LightEA}           & 88.2             & 93.7            & 90.1            & 62.8             & 82.4              & 69.6            & 60.5             & 78.9             & 67.0             & 79.8             & 89.3             & 83.2             & 97.0             & 98.5             & 97.6             & 68.2             & 78.9             & 72.0            \\
                            & \multicolumn{1}{|c|}{\textbf{PipEA(L)}} & \textbf{89.1}    & \textbf{90.9}   & \textbf{90.2}   & \textbf{68.8}    & \textbf{78.9}     & \textbf{72.8}   & \textbf{62.7}    & \textbf{78.1}    & \textbf{67.8}    & \textbf{81.8}    & \textbf{85.6}    & \textbf{83.3}    & \textbf{97.2}    & \textbf{98.2}    & \textbf{97.6}    & \textbf{68.4}    & \textbf{77.6}    & \textbf{72.0}   \\
                            & \multicolumn{1}{|c|}{\textbf{Improv.}}  & \textbf{1.0\%}   & -3.0\% & \textbf{0.1\%}  & \textbf{9.6\%}   & -4.3\%   & \textbf{4.6\%}  & \textbf{3.6\%}   & -1.0\%  & \textbf{1.2\%}   & \textbf{2.48\%}  & \textbf{-4.20\%} & \textbf{0.17\%}  & \textbf{0.21\%}  & -0.27\% & \textbf{0.00\%}  & \textbf{0.29\%}  & -1.65\% & \textbf{0.00\%} \\
\bottomrule
\end{tabular}}
\end{table*}

%The main results of the big table

\subsection{Predicting One-to-one Alignments}
Early EA methods calculate entity pair similarities directly, leading to potential violations of the one-to-one alignment constraint. To circumvent this, the transformation of the EA decoding process into an assignment problem has shown promise \cite{PEEA,mao2022effective}, markedly improving performance:
\begin{equation}
\mathop{\arg\max}\limits_{\mathbf{P}\in \mathbb{P}_{|\mathcal{E}|}} <\mathbf{P}, \Omega>
\end{equation}
Here, $\Omega$ denotes the similarity matrix, and $\mathbf{P}$ is a permutation matrix that outlines the alignment strategy. While the Hungarian algorithm offers a precise solution, its $O(|\mathcal{E}|^3)$ complexity is prohibitive for large KGs. Adopting the Sinkhorn operator \cite{cuturi2013sinkhorn}, we apply a scalable and parallelizable algorithm, significantly reducing the computational load to $O(q|\mathcal{E}|^2)$ with $q$ set to 10 iterations.
\begin{equation}
\begin{aligned}
    \mathop{\arg\max}\limits_{\mathbf{P}\in \mathbb{P}_{|\mathcal{E}|}} <\mathbf{P}, \overline{\Omega}> &= Sinkhorn(\overline{\Omega})\\
    Sink^{(0)}(\overline{\Omega}) &= exp(\overline{\Omega})\\
    Sink^{(q)}(\overline{\Omega}) &= \mathcal{N}_c(\mathcal{N}_r(Sink^{(q-1)}(\overline{\Omega})))\\
    Sinkhorn(\overline{\Omega}) &= \lim_{q\rightarrow \infty} Sink^{(q)}(\overline{\Omega})
\end{aligned}
\end{equation}
$\mathcal{N}_r(\Omega)=\Omega\oslash(\Omega \mathbf{1}_N \mathbf{1}_N^\top)$ and $\mathcal{N}_c(\Omega)=\Omega\oslash(\mathbf{1}_N \mathbf{1}_N^\top \Omega)$ are the row and column-wise normalization operators of a matrix, $\oslash$ represents the element-wise division, and $1_N$ is a column vector of ones.

\subsection{Reducing Complexity}
The PipEA method introduces potential isomorphism propagation to intricately map the inter- and intra-structural nuances of KGs. Despite its comprehensive approach, this technique inherently increases time complexity, primarily due to the computation of $S \in \mathbb{R}^{(n+m)\times (n+m)}$. To mitigate this, we employ strategies to streamline the computational process:

\noindent
\textbf{Sparse threshold:}
Referenced in Section \ref{Matrix Factorization}, a sparse threshold $\delta$ is applied, setting all scores below $\delta$ to zero. This approach filters out insignificant entities, focusing on those most likely to contribute to accurate propagation outcomes. An in-depth analysis of the impact of $\delta$ is presented in our experimental section.

\noindent
\textbf{Low-rank SVD:}
Following observations by \cite{mao2022effective, wang2024gradient} that significant information in $S$ is concentrated in the top singular values, we employ randomized low-rank SVD \cite{ref_article41} (Section.\ref{Matrix Factorization}). This method approximates matrix decomposition, retaining only the top 1\% of singular values, thus reducing both space and time complexity.

\noindent
\textbf{Sinkhorn operator:}
Adopting the Sinkhorn operator \cite{cuturi2013sinkhorn} in prediction, we apply a scalable and parallelizable algorithm, significantly reducing the complexity to $O(q|\mathcal{E}|^2)$ with $q$ set to 10 iterations.

The time complexity analysis are provided in Appendix \ref{complexity}.

\section{Experiments}
\label{experiment}

\subsection{Experimental Settings}
\subsubsection{Datasets.}
To assess PipEA's effectiveness, we turn to the OpenEA benchmark dataset (V2), thoughtfully designed to closely mirror the data distribution found in real knowledge graphs. Our evaluation encompasses two cross-lingual settings (English-to-French and English-to-German), sourced from the multilingual DBpedia, and two monolingual settings (DBpedia-to-Wikidata and DBpedia-to-YAGO), extracted from popular knowledge bases. In each setting, we consider two sizes: one with 15K pairs of reference entities and another with 100K pairs. In contrast to the conventional use of 30\% of seed alignments for training, we adopt weakly supervised scenarios, employing only 1\% of seed alignments randomly sampled from the datasets. Appendix \ref{DatasetsEva} provides comprehensive statistics.

% \subsubsection{Evaluation Metrics.}
% Performance assessment is based on the official Mean Reciprocal Rank (MRR), H@1, and H@10 metrics, widely recognized and embraced in EA studies. Elevated H@1, H@10, and MRR scores signify superior EA performance. Our default alignment direction is from left to right (e.g., EN as the source KG and FR as the target KG for the EN-FR dataset).

\subsubsection{Baselines.}
We compare our PipEA with 10 prominent methods that solely rely on the original structure information:
\begin{itemize}
\item[$\bullet$] Basic models: GCN-Align  \citep{ref_article37}, MRAEA  \citep{ref_article24}, RREA  \citep{ref_article25}, PSR  \citep{ref_article38}, Dual-AMN  \citep{ref_article29}, LightEA \citep{LightEA} and PEEA  \citep{PEEA}.
\item[$\bullet$] Semi-supervised: We also evaluate the PipEA in the \textit{semi-supervised settings}. In addition to the above models, we have added a model that specializes in semi-supervision such as BootEA \citep{ref_article23}, KECG \citep{KECG}, SEA \citep{SEA}.
\end{itemize}
We adhere to the default hyper-parameters as reported in their respective literature. All experiments were diligently executed on a server equipped with an NVIDIA A100 and NVIDIA A800 GPU. 

% \subsubsection{Iterative training strategy.}
% After training the base model with fundamental settings, in every epoch $K_e$ (in our paper $K_e= 10$), cross-KG entity pairs that are mutual nearest neighbors in the vector space are proposed and added to a candidate list $N^{cd}$. An entity pair in $N^{cd}$ is incorporated into the training set if it remains a mutual nearest neighbor for $K_c$ consecutive rounds (where $K_c=10$).

\subsubsection{Parameter Settings.}
Following previous studies\cite{tsitsulin2018verse, STRAP,perozzi2014deepwalk}, we set the embedding dimension to $d = 128$ and determine the token match score $\epsilon = 0.00001$. We employ 8 propagation and refinement iterations ($L_1 = 8$ and $L_2 = 8$). For further insights into hyper-parameters, kindly refer to Appendix \ref{Hyper-parameters}.

% \begin{table}[t]
% \centering
% \caption{Hyperparameter settings of model.}
% \label{tb:hp}
% \Huge
% % \renewcommand\arraystretch{1}
% \setlength{\tabcolsep}{2pt}
%  \resizebox{\linewidth}{!}{
% \begin{tabular}{ccccccccccc}
% \toprule[1.5pt]
% {}   & PSR     & SEA     & MRAEA & KECG  & BootEA & RREA         & Dual-AMN & GCN-Align & PEEA        & Ours        \\
% \midrule
% $d$      & 300,300 & 100     & 50,50 & 128   & 75     & 100,100      & 128,128  & 300       & 128,128,128 & 128 \\
% $lr$     & 0.005   & 0.01    & 0.001 & 0.005 & 0.01   & 0.005        & 0.005    & 0.02      & 0.005       & 0.005       \\
% $b$      & 512     & 5e$^{3}$    & 1e$^{4}$ & -     & 2e$^{4}$  & 1.5e$^{4}$,1e$^{5}$ & 1024     & -         & 1024        & 1024        \\
% $\eta$   & -       & 10      & 5     & 1     & 10     & 5            & -        & 5         & -           & -           \\
% $\gamma$ & 1.5,2.0 & 1.5,2.0 & 1     & 2     & -      & 3            & 1        & 3         & 15          & 15          \\
% $dr$     & 0.3     & -       & 0.3   & -     & -      & 0.3          & 0.3      & 0         & 0.3         & 0.3         \\
% $eps$    & 40      & 2000    & 5000  & 1000  & 500    & 1200         & 100      & 2000      & 60          & 60 \\
% \bottomrule[1.5pt]
% \end{tabular}}
% \end{table}

\begin{figure}[t]
\centering
	\subfloat[PipEA]{\includegraphics[width = 0.45\linewidth]{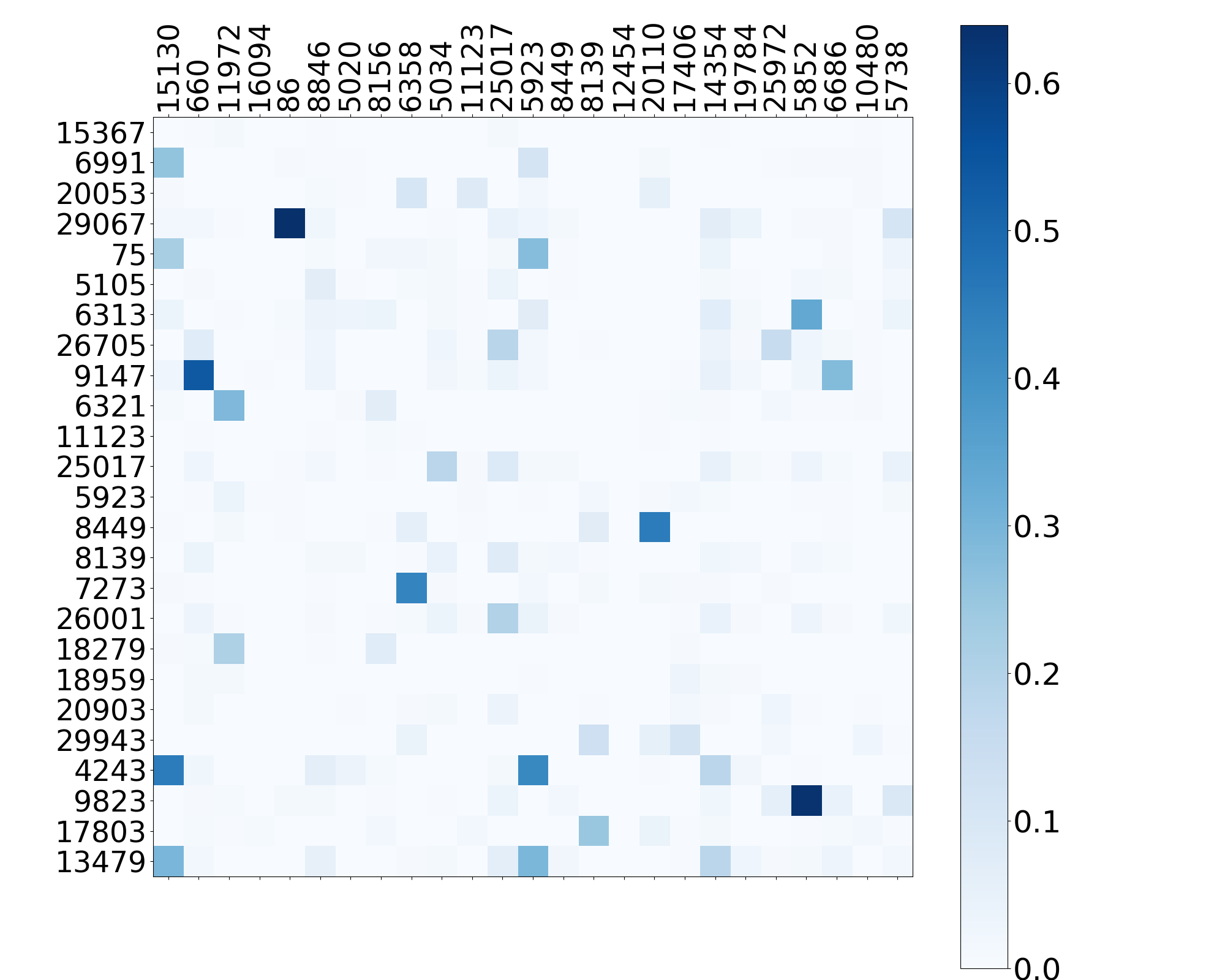}}
	% \hfill
	\subfloat[PEEA]{\includegraphics[width = 0.45\linewidth]{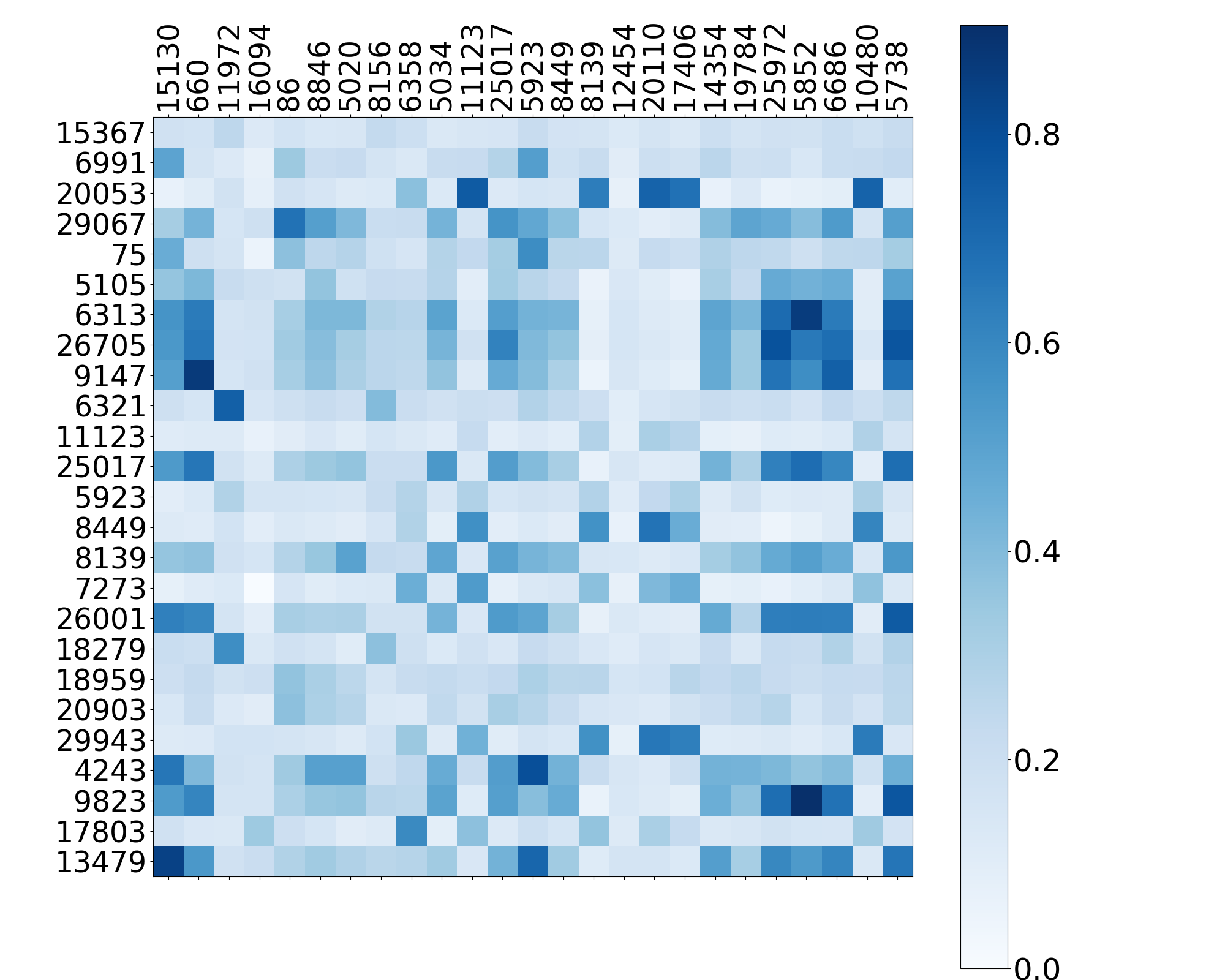}}
\caption{Similarity matrix of our method and PEEA encoder on sub-graph from 15KEN-DE.} \label{fig3}
\vspace{-0.3cm}
\end{figure}

\subsection{Main Results}
\label{main results}
Our evaluation of PipEA across six cross-lingual and monolingual datasets, under the condition of 1\% seed alignments, is summarized in Table \ref{tab1}. The results lead to several key insights:

(1) PipEA significantly outperforms standard alignment models across various encoders in both cross-lingual and monolingual settings. For example, on the 100K datasets, it improves H@1 by 61.3\% and 64.8\% with PEEA. PipEA(D), using DualAMN as the encoder, shows gains of over 50\%, achieving 92.5\% H@1 on the 15K EN-FR dataset. In terms of Mean Reciprocal Rank (MRR), PipEA(D), PipEA(P), and PipEA(L) exhibit improvements ranging from 20.4\% to 52.8\% on 15K datasets and from 15.6\% to 61.3\% on 100K datasets. These results demonstrate PipEA's ability to leverage isomorphic subgraphs and robust cross-graph propagation, making it highly adaptable and effective across different dataset sizes and encoder types.

(2) PipEA benefits significantly from semi-supervised settings, which iteratively incorporate high-quality aligned pairs into training. On the 15K EN-FR and 100K DBP-Wiki datasets, PipEA consistently improves both H@1 and MRR. Even without the iterative strategy, PipEA(D), PipEA(P), and PipEA(L) surpass most iterative models on 15K datasets, though the iterative approach further enhances performance. These results highlight PipEA’s effectiveness in refining models through iterative alignment.

(3) Across all datasets, PipEA outperforms existing models by addressing a key challenge in weakly supervised EA—limited neighborhood propagation within individual graphs. By leveraging isomorphic subgraphs and facilitating cross-graph neighborhood propagation, PipEA significantly improves alignment accuracy.

\subsection{Case Study}
Our case study analysis, focused on the 15KEN-DE dataset, examines the pairwise similarity quality between SOTA models and PipEA. By analyzing a subgraph of 25 entities, we visualized the similarity matrices of PipEA and PEEA through heat maps (Fig.\ref{fig3}). PipEA exhibits markedly clearer and more distinct similarity distributions than PEEA, underscoring its superior ability to refine similarity matrices for accurate entity pair identification. This clarity in similarity distribution explains PipEA's improved H@1 performance contrasted with H@10 performance, indicating its effectiveness in establishing confident one-to-one alignments without necessitating an extensive candidate pool.

\begin{table}[t]
\caption{H@1 of different methods under different supervised settings. PEEA is the encoder.}
\renewcommand\arraystretch{1.2}
\setlength{\tabcolsep}{2pt}
\Huge
\label{ratio}
\resizebox{\linewidth}{!}{
\begin{tabular}{ccccccccccc}
\toprule[1.5pt]
\multicolumn{1}{c}{Dataset}      & \multicolumn{5}{c}{15KEN-DE}                                                      & \multicolumn{5}{c}{15KEN-FR}                                                       \cr
\cmidrule(lr){1-1}\cmidrule(lr){2-6}\cmidrule(lr){7-11}
Seed Ratio & 1\%           & 5\%            & 10\%           & 20\%           & 30\%           & 1\%            & 5\%            & 10\%           & 20\%           & 30\%           \\ \midrule[1.2pt]
Ours           & \textbf{85.4} & \textbf{90.64} & \textbf{92.47} & \underline{94.62}          & \underline{95.46}          & \textbf{58.12} & \textbf{79.87} & \textbf{83.56} & \textbf{88.07} & \textbf{90.24} \\
PEEA           & \underline{68.67}         & \underline{84.58 }         & \underline{91.13}          & \textbf{94.75} & \textbf{95.94} & \underline{44.07}          & \underline{68.73 }         & \underline{77.4 }          & \underline{84.88}          & \underline{90.04}          \\
RREA           & 48.5          & 76.1           & 83.44          & 88.74          & 90.53          & 26.23          & 55.64          & 68.64          & 78.2           & 81.91          \\
Dual-AMN       & 51.95         & 76.46          & 86.13          & 91.2           & 93.17          & 25.28          & 55.51          & 68.11          & 78.78          & 84.11          \\
MRAEA          & 28.63         & 68.56          & 81.42          & 88.1           & 90.72          & 14.42          & 45.82          & 61.75          & 73.82          & 79.76          \\
PSR            & 21.59         & 62.17          & 76.83          & 85.61          & 89.09          & 15.18          & 45.81          & 60.69          & 73.06          & 79.66          \\
GCN-Align      & 10.99         & 20.76          & 24.21          & 27.43          & 30.5           & 3.68           & 12.84          & 18.76          & 23.52          & 28.83          \\ \bottomrule[1.5pt]
\end{tabular}}
\end{table}

\begin{table}[t]
\Huge
\caption{Ablation experiments without the iterative training strategy on cross-lingual and mono-lingual datasets}\label{tab4}
\renewcommand\arraystretch{1.1}
\setlength{\tabcolsep}{2pt}
\resizebox{\linewidth}{!}{
\begin{tabular}{ccccccccccccc}
\toprule[1.5pt]
\multirow{2}{*}{Methods} & \multicolumn{3}{c}{15KEN-DE}                     & \multicolumn{3}{c}{15KEN-FR}  & \multicolumn{3}{c}{15KDBP-Yago}                  & \multicolumn{3}{c}{15KDBP-Wiki}  \cr
\cmidrule(lr){2-4}\cmidrule(lr){5-7}\cmidrule(lr){8-10}\cmidrule(lr){11-13}
& H@1         & H@10        & MRR            & H@1         & H@10        & MRR  & H@1         & H@10        & MRR            & H@1         & H@10        & MRR   \\ \midrule
Ours                     & \textbf{85.40} & \textbf{92.03} & \textbf{87.80}  & \textbf{58.12} & \textbf{74.42} & \textbf{63.69} & \textbf{96.90}  & \textbf{98.48} & \textbf{97.51} & \textbf{75.40} & 82.55          & \textbf{77.94} \\ 
w\textbackslash{}o RS    & 79.99          & 91.24          & 81.24          & 49.72          & 72.72          & 58.95   & 95.32          & 98.2           & 96.47          & 69.89          & 81.57          & 74.51          \\
w\textbackslash{}o IS    & 78.48          & 88.84          & 81.98          & 48.27          & 70.4           & 55.53 & 92.46          & 96.63          & 93.95          & 67.93          & 77.49          & 71.08          \\
w\textbackslash{}o PS    & 75.64          & 88.06          & 79.92          & 49.53          & 70.18          & 56.52 & 92.32          & 97.99          & 94.43          & 67.36          & \textbf{82.72} & 73.07          \\ 
\bottomrule[1.5pt]
\end{tabular}}
\end{table}

\subsection{Discussions for Supervised Settings}
To assess the generality of our method under different supervised settings, we conducted comprehensive experiments comparing PipEA with baseline methods at different training data ratios (1\%, 5\%, 10\%, 20\%, and 30\%) using the 15KEN-DE and 15KEN-FR datasets. The results summarized in Table \ref{ratio} reveal several key insights: (1) Our method demonstrates remarkable performance even in conventional supervised settings. While its advantage over other baselines may diminish as the training ratio increases, PipEA consistently outperforms them. Notably, it achieves superior performance compared to baselines on the 15K EN-FR dataset, even when the seed alignment ratio is as high as 30\%. This phenomenon can be attributed to the nature of aggregation-based EA models, which operate by propagating neighborhood information. PipEA's unique capability to propagate information across graphs extends the size of the neighborhood involved in this propagation process, contributing to its sustained success across varying degrees of supervision. (2) The performance gap between PipEA and baselines in weakly supervised scenarios gradually narrows as the training ratio increases. The alignment performance on the 15K EN-DE dataset is lower than PEEA  \cite{PEEA} when the seed alignment ratio exceeds 10\%. This is because the propagation gap is diminished with substantial seed alignments available and PipEA was originally designed to tackle the challenges of limited seed alignments.

\subsection{Ablation Study}
To comprehensively evaluate the contributions of each component within PipEA, we conducted an ablation study based on PEEA encoder, introducing three variants of the model, as presented in Table \ref{tab4}:
(1) w/o RS (Refinement Scheme): In this variant, we excluded the refinement scheme, directly multiplying the initial similarity matrix $\Omega_0$ with the generated similarity matrix $\Omega_0^\prime$. As discussed in Section \ref{Refining Scheme}, directly fusing different similarity matrices can result in a loss of topological consistency information. Compared to our complete PipEA method, w/o RS led to a reduction in MRR ranging from 1.04\% to 6.56\% on the 15K datasets. Notably, the H@1 exhibited a significant drop of 8.4\% on 15K EN-FR. This suggests that without the refinement scheme, the model may produce more erroneous predictions, especially when only one outcome can be predicted.
(2) w/o IS (Initial Similarity): In this case, we excluded the initial similarity matrix $\Omega_0$ and solely utilized $\Omega_0^\prime$ as the final similarity matrix. It is crucial to note that the initial similarity matrix generated by the aggregation-based EA model primarily captures local similarities arising from its n-hop neighborhood information aggregation. Conversely, our method provides global similarities. By removing these local similarities, we observed a notable decrease in model performance, ranging from 4.44\% to 9.85\% in terms of H@1.
(3) w/o PS (Potential Isomorphism Propagation): In this variant, we omitted the generated similarity matrix $\Omega_0^\prime$. These experiments offered insights into the effectiveness of potential isomorphism propagation. Across the 15K datasets, w/o PS resulted in a substantial reduction in H@1, ranging from 4.58\% to 9.76\%. This outcome underscores the critical role of potential isomorphism propagation in our method. Furthermore, it highlights that our refinement scheme can effectively preserve topological consistency within the similarity matrix, leading to improved alignment accuracy.

\begin{figure}[t]
\centering
\includegraphics[width = \linewidth]{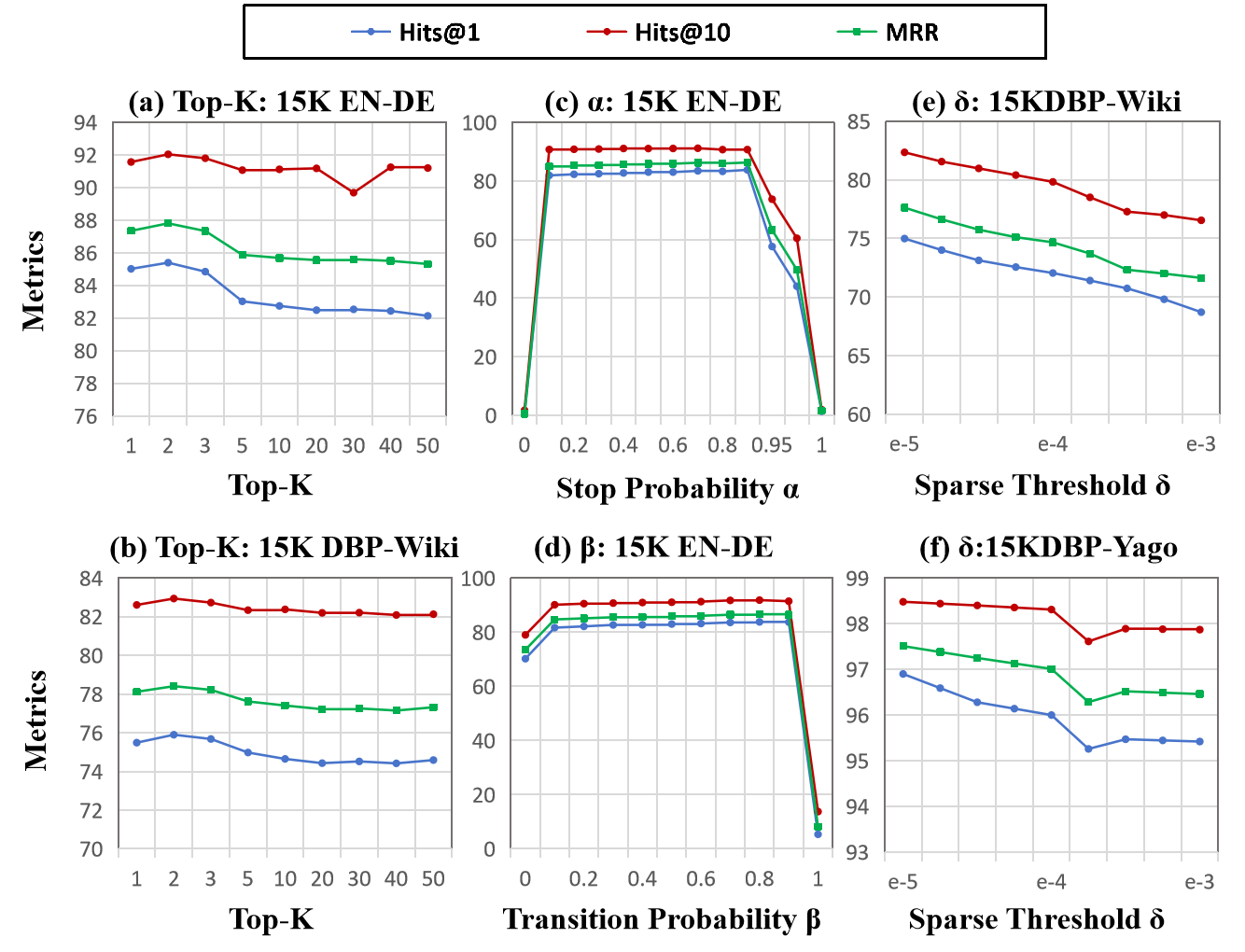}
\caption{Hyper-parameter experiments of our method.} \label{fig2}
\vspace{-0.3cm}
\end{figure}

\subsection{Hyper-parameters}
\label{Hyperparameters}

\subsubsection{Top-K Selection.}
PipEA employs a parameter, denoted as $k$, to select potentially aligned entities for propagation through normalization operations. It is represented as $\varphi_k$ in the normalization operations Eq.\ref{normalization operations}. Fig.~\ref{fig2}(a)(b) illustrates the sensitivity of alignment accuracy to the number of potentially aligned entities. Interestingly, our method achieves peak performance on both the EN-DE and DBP-Wiki datasets across all metrics when $k=2$. Further increasing the value of $k$ leads to diminishing accuracy, suggesting that allowing too many entities to share isomorphic subgraphs within the unified space may negatively impact performance. Consequently, we set $k=2$ as the optimal choice for our study.

\subsubsection{Random Walk Probability.}
The parameter $\alpha$ determines the probability of stopping at the current entity during a random walk. As depicted in Fig.~\ref{fig2}(c), PipEA exhibits robust performance within the range of $\alpha$ from 0.1 to 0.9. The highest H@1 score (85.4\%) is achieved at approximately $\alpha=0.7$, which we adopt in our study. Notably, our experiments reveal that PipEA's effectiveness is compromised under two extreme conditions: when $\alpha=0$, indicating constant random walks among all entities across graphs, and when $\alpha=1$, signifying a lack of neighborhood information propagation. These findings empirically validate the effectiveness of our method.

\subsubsection{Inter-Intra Graph Balance.}
Eq. \ref{construction process} outlines the propagation operator, consisting of intra-graph and inter-graph propagation components, with $\beta$ controlling the balance between these aspects. Specifically, $\beta$ determines the probability of propagation occurring within the intra-graph part compared to the inter-graph part. Fig.~\ref{fig2}(d) demonstrates that PipEA consistently performs well for $\beta$ values ranging from 0.1 to 0.9. We set $\beta=0.5$ in this work. Notably, when $\beta=0$, signifying propagation solely across graphs, H@1 exhibits a slight degradation, affirming PipEA's ability to capture neighborhood information in the unified space. Conversely, when $\beta=1$, indicating exclusive intra-graph propagation, H@1 drops significantly (1.36\%), underscoring the challenge of capturing entity dependencies between different graphs relying solely on single-graph structures. This emphasizes the rationale behind our method, which combines intra-graph and inter-graph propagation.

\subsubsection{Threshold}
To ensure the non-negativity of the generated matrix and its applicability across graphs of varying scales, we utilize the threshold $\delta$, as outlined in Section \ref{Propagation Scheme}. Our experiments indicate that accuracy declines as $\delta$ decreases, signifying the loss of valuable neighborhood information, as depicted in Fig.~\ref{fig2}(e)(f). This reaffirms the effectiveness of PipEA on 100K datasets and its significant performance improvements compared to baselines.

\section{Conclusion}
This research addresses the challenges of weakly supervised EA characterized by limited seed alignments. Through a propagation perspective, we analyze how aggregation-based EA models utilize operators to propagate pairwise entity similarities. A key insight from our theoretical analysis is the existence of isomorphic subgraphs among potentially aligned entities, facilitating effective information propagation essential for EA tasks. Based on it, we introduce the novel propagation operator, which enables cross-graph information propagation, and prove its convergence. We utilize this operator to develop PipEA, a general and novel approach that synergizes intra-graph and inter-graph propagation techniques while refining similarity matrices to enhance alignment accuracy. Our experimental results affirm PipEA's effectiveness, showcasing its superiority over state-of-the-art models in both weakly supervised and fully supervised contexts.

\newpage
\bibliographystyle{ACM-Reference-Format}
\bibliography{sample-base}

\appendix

\section{Notation}
\label{appendix:notation}
The notations used in this paper are summarized in the Table.\ref{symbol}:

\begin{table}[b]
\caption{Symbols and Notations}
\resizebox{0.48\textwidth}{!}{
\begin{tabular}{cc}
\toprule
Notation &Explanation\\
\midrule
$G_s, G_t$ & The source and target KGs to be aligned                                           \\
$\mathcal{E}_s, \mathcal{E}_t$ & The entity set of $G_s, G_t$                                           \\
$\mathcal{R}_s, \mathcal{R}_t$ & The relation set of $G_s, G_t$                                          \\
$\mathcal{T}_1, \mathcal{T}_2$ & The relation triple set of $G_s, G_t$                                          \\
$A_s, A_t$ & The adjacency matrices of $G_s, G_t$                                        \\
$\hat{A}_s, \hat{A}_t$ & The normalized adjacency matrices of $G_s, G_t$  \\
$n, m$ & The node number of $G_s, G_t$
 \\
\midrule
$x_i$ & The entity feature of entities $e_i$       \\
$X_s, X_t$ & The final entity embeddings of graphs $G_s, G_t$    \\
$\Omega(i,j)$ & The pairwise similarity of entities $i$ and $j$    \\
$\textbf{P}$                  & The permutation matrix across $G_s$ and $G_t$ \\ 
$P$                  & The propagation operator \\ 
$\Lambda_3$     &  Potential isomorphism propagation operator                 \\
$[\cdot||\cdot]$               & Vertical concatenation of two column vectors                       \\
$|\cdot|$               & The number of elements in a set                       \\
$\oslash$               &  The element-wise division                      \\
$\odot$               &  The element-wise multiplication                      \\
$\otimes$   &The Kronecker product   \\
$\pi_u(v)$               &  The transition probability from entity $u$ to $v$                      \\
$rank(\cdot)$               &  The rank of matrix  \\
$det(\cdot)$               &  The determinant of the matrix  \\
$Pr(\cdot)$               &  The probability of given event \\
$\mathcal{N}_{e_i}$                     & The neighboring entity set of entity $e_i$                                  \\ 
\hline
\bottomrule
\end{tabular}}
  \label{symbol}
\end{table}

\section{Theoretical results}

\begin{table*}[t]
\caption{Datasets statistics.}\label{Datasets}
\small
\resizebox{!}{!}{
\begin{tabular}{c|cc|cc|cc|cc|cc|cc}
\toprule[1pt]
\multicolumn{7}{c|}{Cross-Lingual Datasets}                       & \multicolumn{6}{c}{Mono-Lingual Datasets}                      \\
\midrule
         & \multicolumn{2}{c}{15KEN-FR} & \multicolumn{2}{c}{15KEN-DE} & \multicolumn{2}{c|}{100KEN-FR} & \multicolumn{2}{c}{15KDBP-Wiki} & \multicolumn{2}{c}{15KDBP-Yago} & \multicolumn{2}{c}{100KDBP-Wiki} \\
\midrule
Relation & 193    & 166    & 169    & 96     & 379     & 287     & 167     & 121      & 72      & 21       & 318      & 239       \\
Triple   & 96318  & 80112  & 84867  & 92632  & 649902  & 561391  & 73983   & 83365    & 68063   & 60970    & 616457   & 588203    \\
Anchors  & 15000  & 15000  & 15000  & 15000  & 100000  & 100000  & 15000   & 15000    & 15000   & 15000    & 100000   & 100000    \\
Train    & 150    & 150    & 150    & 150    & 1000    & 1000    & 150     & 150      & 150     & 150      & 1000     & 1000      \\
Test     & 14850  & 14850  & 14850  & 14850  & 99000   & 99000   & 14850   & 14850    & 14850   & 14850    & 99000    & 99000    \\
\bottomrule[1pt]
\end{tabular}}
\end{table*}

\begin{table}[]
\centering
\caption{Hyperparameter settings of model.}
\label{tb:hp}
 \resizebox{\linewidth}{!}{
\begin{tabular}{cccccccc}
  \toprule
  Models & $d$ & $lr$ & $b$ &$\eta$ &$\gamma$ &$dr$ &$eps$\\
  \midrule
  PSR&300,300&0.005&512&-&1.5,2.0&0.3&40\\
  SEA&100&0.01&5000&10&1.5,2.0&-&2000\\
  MRAEA&50,50&0.001&10000&5&1.0&0.3&5000\\
  KECG&128&0.005&-&1&2.0&-&1000\\
    BootEA&75&0.01&20000&10&-&-&500\\
  RREA&100,100&0.005&15000,100000&5&3.0&0.3&1200\\
  Dual-AMN&128,128&0.005&1024&-&1.0&0.3&100\\
  GCN-Align&300&0.02&-&5&3.0&0&2000\\
  PEEA&128,128,128&0.005&1024&-&15.0&0.3&60\\
  LightEA&1024,1024&-&-&-&-&-&-\\
  Ours&128,128,128&0.005&1024&-&15.0&0.3&60\\
  \bottomrule
\end{tabular}}
\end{table}

\subsection{Proof of Proposition \ref{pro1}}
\label{ProofofProposition1}
Consider the entity representation in an aggregation-based encoder, expressed by Equation \ref{GCNembedding}:
\begin{equation}
    x_i=\sum_{j=1}^{n}{\lambda_{i,j}x_j}
\end{equation}
where $x_i$ represents an entity and $x_j$ represents neighboring entities. The propagation weights $\lambda_{i,j}$ are given by:
\begin{equation}
\label{lambadavalue}
    \lambda_{i,j}=\frac{\mathcal{B}_{(x_i,x_j)\in \mathcal{N}}}{|\mathcal{N}_{x_i}|}
\end{equation}
In the above equation, $\mathcal{B}$ is an indicator function that returns 1 if there exists a relation between $x_i$ and $x_j$ and 0 otherwise. The similarity between entities $x_i \in \mathcal{E}_s$ and $y_i \in \mathcal{E}_t$ can be calculated using the inner product as follows:
\begin{equation}
\label{lambadavalue}
    \omega_{i,j}=x_i\cdot y_j=\sum_{k=1}^n\sum_{l=1}^m\lambda^s_{i,k}\lambda^t_{j,l}x_k\cdot y_l=\sum_{k=1}^n\sum_{l=1}^m\lambda^s_{i,k}\lambda^t_{j,l}\omega_{k,l}
\end{equation}
The above equation illustrates how the similarity between two entities is influenced by their associated neighbors.

We define pairwise entity similarities for the two knowledge graphs as a matrix $\Omega = (\omega_{i,j})^{n,m}_{i=1,j=1}$. Additionally, matrices $\Lambda^s = (\lambda_{i,j}^s)^{n,n}_{i=1,j=1}$ and $\Lambda^t = (\lambda_{i,j}^t)^{m,m}_{i=1,j=1}$ are constructed to serve as propagation operators for the source and target knowledge graphs, respectively.
A key constraint on the propagation matrix $\Lambda$ is introduced in Equation \ref{lambdasum}, ensuring that each row sums to 1:
\begin{equation}
\label{lambdasum}
|\Lambda(i,:)| =\sum_{j=1}^{n}\lambda_{i,j}=\frac{\sum_{j=1}^{n}|{\mathcal{B}{(x_i,x_j)\in \mathcal{N}}}|}{|\mathcal{N}{x_i}|}=1
\end{equation}
This condition signifies that each element $\lambda_{i,j}$ represents the probability of information propagation from entity $x_i$ to $x_j$. Consequently, we establish that $\Lambda^s \Omega (\Lambda^t)^\top = \Omega$, demonstrating that $\Lambda$ functions as a random walk strategy facilitating information propagation within the similarity matrix $\Omega$. The iterative application of $\Lambda$ leverages its properties related to fixpoints, essential for effective alignment.

\subsection{Proof of Proposition \ref{pro2}}
\label{ProofofProposition2}
We aim to show that ${rank(\Lambda \Lambda^\top ) < n}$, the complementary events, have zero probability. Let's take $Y = \Lambda \Lambda^\top$ as an example. It follows that $rank(Y) < n$ if and only if $rank(Y^\top Y) < n$. This occurs if and only if $det(Y^\top Y) =0$. Therefore, we need to demonstrate that the probability of $det(Y^\top Y) =0$ is zero. We have
\begin{equation}
g(Y) = det(Y^\top Y)
\end{equation}
where $g$ is a polynomial. We are interested in the probability of its zero-level set. However, the zero-level set of a non-zero polynomial has zero Lebesgue measure. In our case, $g$ is not the constant zero function since $g(Q) = 1$. Consequently, we have
\begin{equation}
\label{pro2equ}
Pr(g(Y) = 0) =\int_{g^{-1}({0})}^{}f(x)dx = 0
\end{equation}
where the integral is zero since we integrate over a set with zero measure. $Pr$ represents the probability of an event occurring, and $f(x)$ is the probability density function of $x$. The meaning of the above equation is that the probability of the event ${det(Y> Y) = 0}$ is zero, Thus we prove that 
\begin{equation}
\label{rank(Y)}
    rank(Y) = n
\end{equation}
Theorem 3.4 in \cite{EASIM} shows that embedding-based EA models focus on identifying fixed-point sets of pairwise entity similarities through an alignment mapping function, aimed at aligning KG matrices. In our analysis, these matrices, particularly the similarity matrix $\Omega$, are treated as edge weights in mathematical graph matching. As demonstrated in the proof of Proposition \ref{pro1}, with $\Lambda^s\Omega(\Lambda^t)^\top = \Omega$ and the reconstruction defined by $Y = \Lambda\Lambda^\top$, we confirm that each entity corresponds to a match in another graph, indicating isomorphic subgraphs across the unified space.

\subsection{Proof of Proposition \ref{pro3}}
\label{ProofofProposition3}
We aim to prove the proposition by satisfying the assumptions of \textit{Theorem 5.1}  \cite{saad2011numerical}. Although the theorem establishes convergence to the Schur vectors $[q_1,\ldots, q_k]$ associated with $\lambda_1, . . . , \lambda_k$, it's important to note that for symmetric matrices, these Schur vectors coincide with the $k$ orthogonal eigenvectors.

Let $\Omega=[\Omega_1,\ldots,\Omega_k]$. Our objective is to show that 
\begin{equation}
    rank(P_i[\Omega_1,\ldots,\Omega_k]) = i
\end{equation}
where $1 \leq i \leq k$, and $P_i$ represents the spectral projector associated with eigenvalues $\lambda_1, \ldots, \lambda_i$. Given the distinctness and simplicity of the eigenvalues (by assumption) and the symmetry of the operator (implying the equivalence of left and right eigenvectors), we can express the spectral projector as:
\begin{equation}
P_i = \sum_{j=1}^i q_jq_j^\top
\end{equation}
Here, we assume without loss of generality that $q_j$ has a unitary norm. Equivalently, we can write $P_i = Q_iQ_i^\top$, where $Q_i = [q_1,\ldots , q_i]$. Consequently, we need to demonstrate that
\begin{equation}
    rank(Q_iQ_i^\top[\Omega_1,\ldots,\Omega_i]) = i
\end{equation}
where $1 \leq i \leq k$. This occurs with probability 1 according to Proposition \ref{pro2}, as shown earlier.

\section{Datasets}
\label{DatasetsEva}

In our experimental evaluation, we employed six datasets curated within the OpenEA framework  \cite{DBP2}. These datasets encompass two cross-lingual settings, namely English-to-French and English-to-German, extracted from the multi-lingual DBpedia. Additionally, we considered two mono-lingual settings derived from well-established knowledge bases, specifically DBpedia-to-Wikidata and DBpedia-to-YAGO. A comprehensive summary of the dataset statistics is provided in Table \ref{Datasets}.

\section{Time complexity}
\label{complexity}
To streamline our analysis, we assume both graphs have 'n' nodes.
The propagation process involves two main steps: propagation computation and matrix factorization. The computation of propagation is a pivotal step. Its time complexity is bounded by $\mathcal{O}(m \cdot n \cdot L_1)$, where $m$ represents the number of edges, and $L_1$ signifies the count of propagation iterations. In the matrix factorization step, extensive calculations are performed on a dense matrix of size $\mathcal{O}(n^2)$. Consequently, the time complexity is $\mathcal{O}(m \cdot n \cdot L_1)$, mainly due to the computational cost of the former operation. To mitigate this, we introduce an efficient generalized push algorithm, which reduces the complexity to $\mathcal{O}\left(\frac{S}{\delta}\right)$, with $\delta$ controlling the trade-off between computational expense and proximity matrix sparsity.

The refinement process consists of $L_2$' iterations. In each iteration, our method computes the left and right multiplication of a dense $n\times n$ matching matrix with two adjacency matrices from graphs, each with an average degree denoted as $\overline{d_1},\overline{d_2}$, respectively. Consequently, the time complexity for this update step is $O(n^2(\overline{d_1}+\overline{d_2}))$. Normalizing the matrix in each iteration and incorporating token match scores requires $O(n^2)$ time. Therefore, the overall time complexity stands at $O(L_2n^2(\overline{d_1}+\overline{d_2}))$.

\section{Hyper-parameters}
\label{Hyper-parameters}
Our experimental configuration is based on the default optimal settings of various hyperparameters. Specifically, we focus on the following hyperparameters: the number of epochs ($eps$), learning rates ($lr$), batch sizes ($b$), negative sampling rates ($\eta$), dropout rates ($dr$), and the number of GNN layers ($L$). Our primary focus during experimentation involves a grid search over embedding dimensions ($d$) and margins ($\gamma$). Additionally, we maintain consistency with the original papers when it comes to the number of multi-head attention mechanisms used in models like MRAEA, RREA, and Dual-AMN, which are fixed at two. To conduct our experiments in a rigorous manner, we vary the embedding dimension ($d$) within the range of $(50, 100, 128, 300)$ and the margin ($\gamma$) within the range of $(0, 0.5, 1.0, 1.5, 2.0, 2.5, 3.0, 10.0, 15.0, 20.0)$. It is essential to note that, for the sake of fair comparison, we exclude attributes when employing the GCN-Align method. A comprehensive summary of all hyperparameters employed in our experiments is presented in Table~\ref{tb:hp}.

\end{document}